\documentclass[12pt]{article}
\usepackage{amsfonts}
\usepackage{amsmath}
\usepackage{epsfig}
\usepackage{latexsym}
\usepackage{amssymb}
\usepackage{graphicx}
\usepackage{color}

\newtheorem{theorem}{Theorem}

\newtheorem{corollary}[theorem]{Corollary}
\newtheorem{lemma}[theorem]{Lemma}
\newtheorem{proposition}[theorem]{Proposition}
\newtheorem{remark}[theorem]{Remark}
\newenvironment{proof}[1][Proof]{\noindent\textbf{#1.} }{\ \rule{0.5em}{0.5em}}

\begin{document}

\title{Majorization and additivity for multimode bosonic Gaussian channels}

\author{V. Giovannetti$^{1}$, A. S. Holevo$^{2}$, A. Mari$
^{1}$ \\
\\
{\small $^{1}$NEST, Scuola Normale Superiore and Istituto Nanoscienze-CNR, }
\\
{\small I-56127 Pisa, Italy,} \\
{\small $^2$Steklov Mathematical Institute, RAS, Moscow, Russia}}

\date{}
\maketitle

\begin{abstract}
In the present paper the multimode extension of the majorization theorem for bosonic Gaussian channels is obtained,
in particular, sufficient conditions are given, under which the Glauber's
coherent states are the only minimizers for concave functionals of the output state of such a channel. Direct implications
of this multimode majorization result to the positive solution of the famous additivity
problem in the case of Gaussian channels are discussed. In particular, the additivity of the output R\'{e}nyi entropies of arbitrary
order $p>1$ is demonstrated. Finally, an alternative, more direct derivation is given for a majorization property of the Husimi function established by Lieb and Solovej.
\end{abstract}

\section{Introduction}

Recently, the longstanding \textit{Gaussian
optimizer conjecture} in quantum information theory was proven for the class of bosonic Gaussian
gauge-covariant or contravariant channels \cite{P-2}. The conjecture says that
the minimum output entropy of a bosonic Gaussian channel is attained on the
vacuum state (as well as on any coherent state). In \cite{P-1} this result was strengthened
for one-mode channels by establishing that the output for the vacuum or
coherent input  \textit{majorizes} the output for
any other input, in that it minimizes a broad class of concave functionals
of the output states. For a detailed discussion of motivation and applications of these advances to quantum optics and communications we refer to
\cite{P-2}, \cite{P-1}.

In the present paper we obtain further results in this direction. In
Sec. \ref{mlt} we give the multimode extension of the result of \cite{P-1},
and in particular, a precise formulation of sufficient conditions under which the
coherent states are \textit{the only} minimizers. We also discuss direct implications
of this multimode majorization result to the positive solution of another famous conjecture, namely the \textit{additivity
problem} for the Gaussian channels. In particular, we demonstrate the additivity of the output R\'{e}nyi entropies of arbitrary
order $p>1$, which generalizes a result of Giovannetti and Lloyd \cite{gl} for
integer $p$ and special channels.

In Sec. \ref{maj} we generalize a majorization result of Lieb and Solovej
\cite{ls}, basing on the method of the work \cite{P-2}. Wehrl \cite{wehrl}
introduced the \textquotedblleft classical entropy\textquotedblright\ of a
quantum state $\rho $ by the formula
\begin{equation*}
S_{cl}(\rho )=-\int_{\mathbb{C}^{s}}\langle z|\rho |z\rangle \log \langle
z|\rho |z\rangle \frac{d^{2s}z}{\pi ^{s}},
\end{equation*}
where $\langle z|\rho |z\rangle $ is the Husimi function, $|z\rangle $ are
Glauber's coherent vectors, $s$ -- number of the modes. Lieb \cite{lieb}
used exact constants in the Hausdorff-Young inequality (Fourier transform)
and Young inequality (convolution) to prove Wehrl's conjecture \cite{wehrl}: \textit{$
S_{cl}(\rho )$ is minimized by any coherent state $\rho =|\zeta \rangle
\langle \zeta |.$} Recently Lieb and Solovej \cite{ls} gave another
derivation based on the limit version of a similar result for Bloch spin
coherent states. Moreover, in this way they could establish the majorization
property of Glauber's coherent states. In Sec. \ref{maj} we suggest yet a different
(and perhaps most natural) approach to the proof of this property and its
generalization motivated by the recent solution of the Gaussian optimizers
problem \cite{P-2}.

\section{Majorization for gauge-covariant channels}

\label{mlt}

We start with repeating some definitions and notations from \cite{P-2},
restricting to the case of channels with identical input and output spaces.
Consider $s-$ dimensional complex Hilbert space $\mathbf{Z}$ which can be
considered as $2s-$dimensional real space equipped with the symplectic form $
z,z^{\prime }\rightarrow \mathrm{Im\,}z^{\ast }z^{\prime }.$ We will
consider vectors in $\mathbf{Z}$ as $s-$dimensional complex column vectors,
in which case (complex-linear) operators in $\mathbf{Z}$ are represented by
complex $s\times s-$matrices, and $^{\ast }$ denotes Hermitian conjugation.
The gauge group acts in $\mathbf{Z}$ as multiplication by $e^{i\phi },$
where $\phi $ is real number called phase. The Weyl quantization is
described by the unitary displacement operators $D(z)$ acting irreducibly in
the representation space $\mathcal{H}$ and satisfying the canonical
commutation relation
\begin{equation}
D(z)D(z^{\prime })=\exp \left( -i\mathrm{Im\,}z^{\ast }z^{\prime }\right)
D(z+z^{\prime }).  \label{defDD}
\end{equation}
Introducing the annihilation-creation operators of the system $a_{j}$, $
a_{j}^{\dag };\; j=1,\dots ,s$, which satisfy the commutation relations $\left[
a_{j\,,}a_{k}^{\dag }\right] =\delta _{jk}I$, the operator $D(z)$ can be
expressed as
\begin{equation}
D(z)=\exp \sum_{j=1}^{s}\left( z_{j}a_{j}^{\dag }-\bar{z}_{j}a_{j}\right) .
\label{displacement}
\end{equation}

The gauge group has the unitary representation $\phi \rightarrow U_{\phi
}=e^{i\phi N}$ in $\mathcal{H}$ where $N=\sum_{j=1}^{s}a_{j}^{\dag }a_{j}$
is the total number operator. The representation of the gauge group in $
\mathcal{H}$ acts according to the relation $U_{\phi }^{\ast }D(z)U_{\phi
}=D(\mathrm{e}^{i\phi }z),\,\phi \in \lbrack 0,2\pi ].$ A state $\rho $ is
then said to be gauge-invariant if it commutes with all $U_{\phi }$ or,
equivalently, if its characteristic function $\phi (z)=\mathrm{Tr}\rho D(z)$
is invariant under the action of the gauge group. In particular Gaussian
gauge-invariant states are given by the characteristic function of the form
\begin{equation}
\phi (z)=\exp \left( -z^{\ast }\alpha z\right) ,  \label{gausstate}
\end{equation}
where $\alpha $ is a complex correlation matrix satisfying $\alpha \geq
I/2,\,I$ -- the unit $s\times s-$matrix. The vacuum state $|0\rangle \langle
0|$ corresponds to $\alpha =I/2$.

A channel $\Phi $ in $\mathcal{H}$ \ is completely positive trace preserving
map of the Banach space of trace-class operators in $\mathcal{H}$, see, e.g. \cite{H-SCI} for detail. The
channel is called \emph{gauge-covariant} if
\begin{equation}
\Phi \lbrack U_{\phi }\rho U_{\phi }^{\ast }]=U_{\phi }\Phi \lbrack \rho
]U_{\phi }^{\ast }.  \label{g-c}
\end{equation}
In the Heisenberg picture, a bosonic Gaussian gauge-covariant channel $\Phi $
\cite{P-2} is described by the action of its adjoint $\Phi ^{\ast }$ onto
the displacement operators as follows:
\begin{equation}
\Phi ^{\ast }[D(z)]=D(K^{\ast }z)\exp \left( -z^{\ast }\mu z\right) ,
\label{defprima}
\end{equation}
where $K$ is a complex matrix and $\mu $ is Hermitian matrix satisfying the
inequality
\begin{equation}
\mu \geq \pm \frac{1}{2}\left( I-KK^{\ast }\right) .  \label{ineq}
\end{equation}

The gauge-covariant channel is \emph{quantum-limited} if $\mu $ is a minimal
solution of the inequality (\ref{ineq}). Special cases of the maps~(\ref
{defprima}) are provided by the attenuator and amplifier channels,
characterized by matrix $K$ fulfilling the inequalities, $KK^{\ast }\leq I$
and $KK^{\ast }\geq I$, respectively. We are particularly interested in \emph{
quantum-limited attenuator} which corresponds to
\begin{equation}
KK^{\ast }\leq I,\qquad \qquad \mu =\frac{1}{2}\left( I-KK^{\ast }\right) ,
\label{mindef1}
\end{equation}
and \emph{quantum-limited amplifier}
\begin{equation}
KK^{\ast }\geq I,\qquad \qquad \mu =\frac{1}{2}\left( KK^{\ast }-I\right) .
\label{mindef2}
\end{equation}
These channels are diagonalizable: by using  singular value decomposition $
K=V_{B}K_{d}V_{A},$ where $V_{A},V_{B}$ are unitaries and $K_{d}$ is
a diagonal matrix with nonnegative values on the diagonal, we have $KK^{\ast
}=V_{B}K_{d}K_{d}{}^{\ast }V_{B}^{\ast },$ and
\begin{equation}
\Phi \lbrack \rho ]=U_{B}\Phi _{d}[U_{A}\rho U_{A}^{\ast }]U_{B}^{\ast },
\label{unieq}
\end{equation}
where $\Phi _{d}$ is a tensor product
\begin{equation}
\Phi _{d}=\otimes _{j=1}^{s}\Phi _{j}  \label{product}
\end{equation}
of one-mode quantum-limited channels defined by the matrix $K_{d}$ and where
$U_{A}$, $U_{B}$ are the canonical unitary transformations acting on $
\mathcal{H}$, such that
\begin{equation*}
U_{B}^{\ast }D(z)U_{B}=D(V_{B}^{\ast }z),\qquad \qquad U_{A}^{\ast
}D(z)U_{A}=D(V_{A}^{\ast }z),
\end{equation*}
(notice that $U_{A}|0\rangle =|0\rangle ,U_{B}|0\rangle =|0\rangle $).

\begin{theorem}
\label{T1} (i) Let $\Phi $ be a Gaussian gauge-covariant channel, and let $f$
be a concave function on $[0,1],$ such that $f(0)=0,$ then
\begin{equation}
\mathrm{Tr}f(\Phi \lbrack \rho ])\geq \mathrm{Tr}f(\Phi \lbrack |\zeta
\rangle \langle \zeta |])=\mathrm{Tr}f(\Phi \lbrack |0\rangle \langle 0|]).
\label{main1}
\end{equation}
for all states $\rho $ and any coherent state $|\zeta \rangle \langle \zeta
| $ (the value on the right is the same for all coherent states by the
unitary covariance property of a Gaussian channel \cite{H-SCI}).

(ii) If $f$ is strictly concave, and the channel $\Phi $ satisfies one of
the two conditions:

a) $K$ is invertible and\footnote{
For Hermitian matrices $M,N,$ the strict inequality $M>N$ means that $M-N$
is positive definite.}
\begin{equation}
\mu >\frac{1}{2}\left( KK^{\ast }-I\right) ;  \label{mugr}
\end{equation}

b) $KK^{\ast}>I$ and $\mu = \frac{1}{2}\left(KK^{\ast }- I\right)$ (hence $
\Phi$ is a quantum-limited amplifier),

then the equality in (\ref{main1}) is attained only when $\rho $ is a coherent
state.
\end{theorem}

In the case of one mode such a result was obtained in \cite{P-1}.
Our goal here is to generalize it to the case of many modes, in particular by
making precise the conditions in the statement (ii).

\begin{proof}
 (i) By concavity it is sufficient to prove (\ref{main1}) for pure
states $\rho =|\psi \rangle \langle \psi |$. As shown in \cite{P-2} (see
also Appendix 1, Proposition \ref{prop1}), any gauge-covariant channel can
be represented as concatenation $\Phi =\Phi _{2}\circ \Phi _{1}$ of
quantum-limited attenuator $\Phi _{1}$ with operator $K_{1}$ and
quantum-limited amplifier $\Phi _{2}$ with operator $K_{2}.$ Then an argument
similar to \cite{P-1} shows that it is sufficient to prove (\ref{main1})
only for the amplifier $\Phi _{2}.$ Indeed, assume that for any state vector
$|\psi \rangle $
\begin{equation}
\mathrm{Tr}f(\Phi _{2}[|\psi \rangle \langle \psi |])\geq \mathrm{Tr}f(\Phi
_{2}[|0\rangle \langle 0|]).  \label{amp}
\end{equation}
Consider the spectral decomposition $\Phi _{1}[|\psi \rangle \langle \psi
|]=\sum_{j}p_{j}|\phi _{j}\rangle \langle \phi _{j}|,$ where $p_{j}>0,$ then
\begin{eqnarray}
\mathrm{Tr}f(\Phi \lbrack |\psi \rangle \langle \psi |]) &=&\mathrm{Tr}
f(\Phi _{2}[\Phi _{1}[|\psi \rangle \langle \psi |]])  \label{decomp} \\
&\geq &\sum_{j}p_{j}\mathrm{Tr}f(\Phi _{2}[|\phi _{j}\rangle \langle \phi
_{j}|])  \label{decompa} \\
&\geq &\mathrm{Tr}f(\Phi _{2}[|0\rangle \langle 0|])  \label{decomb} \\
&=&\mathrm{Tr}f(\Phi _{2}[\Phi _{1}[|0\rangle \langle 0|]])=\mathrm{Tr}
f(\Phi \lbrack |0\rangle \langle 0|]),  \label{decomc}
\end{eqnarray}
because vacuum is an invariant state of a quantum-limited attenuator.

Let us now prove (\ref{amp}). Since
\begin{equation*}
\min_{\rho }\mathrm{Tr}f(\Phi _{2}[\rho ])=\min_{\rho }\mathrm{Tr}
f(U_{B}\Phi _{d}[U_{A}^{\ast }\rho U_{A}]U_{B}^{\ast })=\min_{\rho }\mathrm{
Tr}f(\Phi _{d}[\rho ]),
\end{equation*}
it is sufficient to consider the diagonal amplifier. The proof for one-mode
quantum-limited amplifier is based on the fact that the complementary
channel has the representation (also based on Proposition \ref{prop1} \cite
{P-2})
\begin{equation}
\tilde{\Phi}_{2}=T\circ \Phi _{2}\circ \Phi _{1}^{\prime },  \label{rep}
\end{equation}
where $\mathrm{T}$ is transposition defined by the relation {$\mathrm{T}
[D(z)]=D(-\bar{z})$ $\,$(}$\bar{z}$ is the complex conjugate vector), and $
\Phi _{1}^{\prime }$ is another quantum-limited attenuator defined by the
operator $K'_{1}=\sqrt{I-K_{2}^{-2}}$. But for a diagonal multimode
amplifier the expression for the complementary channel and also
representation (\ref{rep}) (with diagonal $\Phi _{1}^{\prime }$) follows
from the results for each mode.

The representation (\ref{rep}) implies that
nonzero spectra of the density operators $\Phi _{2}[\rho ]$ and $\Phi
_{2}\circ \Phi _{1}^{\prime }[\rho ]$ coincide for pure inputs $\rho = |\psi \rangle \langle \psi |$ \cite
{P-2}. Then similarly to (\ref{decomp})-(\ref{decompa})
\begin{eqnarray}
\mathrm{Tr}f(\Phi _{2}[|\psi \rangle \langle \psi |]) &=&\mathrm{Tr}f(\Phi
_{2}[\Phi _{1}^{\prime }[|\psi \rangle \langle \psi |]])  \label{decomp1} \\
&\geq &\sum_{j}p_{j}^{\prime }\mathrm{Tr}f(\Phi _{2}[|\phi _{j}^{\prime
}\rangle \langle \phi _{j}^{\prime }|]),  \notag
\end{eqnarray}
where
\begin{equation}
\Phi _{1}^{\prime }[|\psi \rangle \langle \psi |]=\sum_{j}p_{j}^{\prime
}|\phi _{j}^{\prime }\rangle \langle \phi _{j}^{\prime }|,\,p_{j}^{\prime
}>0,  \label{spd}
\end{equation}
is the spectral decomposition of the output of the quantum-limited
attenuator $\Phi _{1}^{\prime }.$ Assume for a moment that $f$ is strictly
concave, then one arrives to the conclusion that for any pure minimizer $
\rho =|\psi \rangle \langle \psi |$ of $\ \mathrm{Tr}f(\Phi _{2}[|\psi
\rangle \langle \psi |])$ the sum (\ref{spd}) necessarily contains only one term, i.e.
\begin{equation}
\Phi _{1}^{\prime }[|\psi \rangle \langle \psi |]=|\phi ^{\prime }\rangle
\langle \phi ^{\prime }|.  \label{pure}
\end{equation}
Indeed, otherwise by the strict concavity the inequality in (\ref{decomp1}) is strict, contradicting
the assumption that $|\psi \rangle \langle \psi |$ is a minimizer of $\
\mathrm{Tr}f(\Phi _{2}[|\psi \rangle \langle \psi |])$ (strict concavity of $
f$ \ also excludes non-pure minimizers). Next, we first consider the amplifier
with $K_{2}>I,$ then the associated attenuator $\Phi _{1}^{\prime }$ is defined by the
operator $K'_{1}=\sqrt{I-K_{2}^{-2}},$ such that $0<K'_{1}<I.$
We then apply the following

\begin{lemma}
\label{2} \ Let $\Phi _{1}^{\prime }$ be a diagonal quantum-limited
attenuator defined by the operator $K'_{1},$ such that $0<K'_{1}<I.$ Then (\ref{pure}) implies that $|\psi \rangle \langle \psi |$ is a
coherent state.
\end{lemma}

For one mode, this is Lemma 2 from \cite{P-1} which implies that any pure
input $\rho $, such that $\Phi _{1}^{\prime }[\rho ]$ is also a pure state, is a
coherent state. The proof is based on the explicit expression for the
complementary channel $\widetilde{\Phi _{1}^{\prime }}$. By using this
expression for each mode, one can generalize the proof to the case of the
diagonal multimode channel $\Phi _{1}^{\prime }.$

This proves (\ref{amp}) for strictly concave $f$ and for the amplifiers $\Phi_2$ with $
K_{2}>I$. An arbitrary concave $f$ can then be monotonically approximated by
strictly concave functions by setting $f_{\varepsilon }(x)=f(x)-\varepsilon
x^{2},\,$\ and passing to the limit $\varepsilon \downarrow 0$ in (\ref{amp}).

In the case of the diagonal amplifier $\Phi _{2}$ with $K_{2}\geq I,$ we
take any sequence of diagonal operators $K^{(n)}>I,\,K^{(n)}\rightarrow
K_{2},$ and consider the corresponding diagonal amplifiers $\Phi _{2}^{(n)}$.
Then $\left\Vert \Phi _{2}^{(n)}[\rho ]-\Phi _{2}[\rho ]\right\Vert
_{1}\rightarrow 0$ and $\mathrm{Tr}f(\Phi _{2}^{(n)}[\rho ])\rightarrow
\mathrm{Tr}f(\Phi _{2}[\rho ])$ for any concave polygonal function $f$ on $
[0,1],$ such that $f(0)=0.$ This follows from the fact that any such
function is Lipschitz, $|f(x)-f(y)|\leq \varkappa |x-y|$, hence $\left\vert
\mathrm{Tr}f(\Phi _{2}^{(n)}[\rho ])-\mathrm{Tr}f(\Phi _{2}[\rho
])\right\vert \leq \varkappa \left\Vert \Phi _{2}^{(n)}[\rho ]-\Phi
_{2}[\rho ]\right\Vert _{1}.$ This implies that (\ref{amp}) holds for
polygonal concave functions $f$ and all quantum-limited amplifiers, hence by
(\ref{decomb}) the inequality (\ref{main1}) with such $f$\ holds for for all
Gaussian gauge-covariant channels. For arbitrary concave $f$ on $[0,1]$
there is a monotonously nondecreasing sequence of concave polygonal
functions $f_{m}$ converging to $f$ pointwise. Passing to the limit $
m\rightarrow \infty $ gives the first statement.

(ii) a) Notice that the conditions on the channel $\Phi $ imply that in the
decomposition $\Phi =\Phi _{2}\circ \Phi _{1}$ the attenuator $\Phi _{1}$ is
defined by the operator $K_{1}$ such that $0<K_{1}^{\ast }K_{1}<I$ (see
Appendix 1, Remark \ref{rem}). Applying the argument involving the relations
(\ref{decomp1}) with strictly concave $f$ to the relations
(\ref{decomp})-(\ref{decomc}), we obtain that for any pure minimizer $\rho =|\psi \rangle
\langle \psi |$ of $\mathrm{Tr}f(\Phi \lbrack |\psi \rangle \langle \psi
|])$ the output of the quantum-limited attenuator $\Phi _{1}[|\psi \rangle
\langle \psi |]$ is necessarily a pure state. Applying Lemma \ref{2} to the
attenuator $\Phi _{1}$ we conclude that $|\psi \rangle \langle \psi |$ is
necessarily a coherent state.

b) In this case we just apply the argument involving the relations (\ref
{decomp1}) with strictly concave $f$ to the quantum-limited amplifier $\Phi
=\Phi _{2}.$
\end{proof}

Theorem \ref{T1} can be extended to Gaussian gauge-contravariant channel
satisfying $\Phi \lbrack U_{\phi }\rho U_{\phi }^{\ast }]=U_{\phi }^{\ast
}\Phi \lbrack \rho ]U_{\phi }$ instead of (\ref{g-c}). The proof follows
from the fact that the complementary $\tilde{\Phi}_{2}$ of the diagonal
quantum-limited amplifier $\Phi _{2}$ is just the diagonal quantum-limited
gauge-contravariant channel (see \cite{P-2} for detail).

\section{Implications for the additivity}

For any $p>1$ the output purity of a channel $\Phi $ is defined as
\begin{equation*}
\nu _{p}(\Phi )=\sup_{\rho \in \mathfrak{S}(\mathcal{H})}\mathrm{Tr}\Phi
\lbrack \rho ]^{p}.
\end{equation*}

\begin{corollary}
\label{dec copy(1)} For any Gaussian gauge-covariant channel $\,\Phi $ the
output purity is equal to $\nu _{p}(\Phi )=\mathrm{Tr}\Phi \lbrack |0\rangle
\langle 0|]^{p}.$ The multiplicativity property
\begin{equation}
\nu _{p}(\Phi \otimes \Psi )=\nu _{p}(\Phi )\nu _{p}(\Psi )  \label{mul-rel}
\end{equation}
holds for any two Gaussian gauge-covariant channels $\Phi $ and $\Psi $.
\end{corollary}

\begin{proof}
The first statement follows from Theorem \ref{T1} by taking $
f(x)=-x^{p}, $ so that $\nu _{p}(\Phi )=-\min_{\rho }\mathrm{Tr}f(\Phi
\lbrack \rho ]).$ The second statement then follows from the fact that the
channel $\Phi \otimes \Psi $ is also gauge-covariant and from
multiplicativity of the vacuum state.
\end{proof}

The output purity for the channel (\ref{defprima}) can be explicitly
computed as
\begin{equation*}
\nu _{p}(\Phi )=\det \left[ \left( \mu +KK^{\ast }/2+I/2\right) ^{p}-\left(
\mu +KK^{\ast }/2-I/2\right) ^{p}\right] .
\end{equation*}
The formula follows from the fact that the state $\Phi \lbrack |0\rangle
\langle 0|]$ is Gaussian with the covariance matrix $\mu +KK^{\ast }/2$ and
from the expression for the spectrum of a Gaussian density operator \cite{hir}.

The minimal output R\'enyi entropy of a channel $\Phi $ is expressed via its
output purity as follows
\begin{equation*}
\check{R}_{p}(\Phi )=\frac{1}{1-p}\log \nu _{p}(\Phi )
\end{equation*}
and multiplicativity property (\ref{mul-rel}) can be rewritten as the
additivity of the minimal output R\'enyi entropy
\begin{equation}
\check{R}_{p}(\Phi \otimes \Psi )=\check{R}_{p}(\Phi )+\check{R}_{p}(\Psi ).
\label{add-mre}
\end{equation}
In the limit $p\downarrow 1$ (or taking $f(x)=-x\log x$) we recover the
additivity of the minimal output von Neumann entropy established in \cite
{P-2}.
\begin{equation*}
\min_{\rho _{12}}H(\left( \Phi \otimes \Psi \right) [\rho _{12}])=\min_{\rho
_{1}}H(\Phi \lbrack \rho _{1}])+\min_{\rho _{2}}H(\Phi \lbrack \rho _{2}]).
\end{equation*}
The additivity result in \cite{P-2} is more general in that it allows the
case where one of the channels is gauge-covariant, while the other is
contravariant. On the other hand, the proof in \cite{P-2} is restricted to
states with finite second moments, while the present one does not require
this.

\section{\protect\bigskip Majorization for quantum-classical Gaussian channel
}

\label{maj}

It is helpful to consider the map $\rho \rightarrow \langle z|\rho |z\rangle
$ as a \textquotedblleft quantum-classical Gaussian
channel\textquotedblright\ which transforms Gaussian density operators into
Gaussian probability densities. We will consider a more general
transformation:
\begin{equation*}
\rho \rightarrow p_{\rho }(z)=\mathrm{Tr}\rho D(z)\rho _{0}D(z)^{\ast },
\end{equation*}
where $D(z)$ are the displacement operators, $\rho _{0}$ is the Gaussian
gauge-invariant state with the quantum characteristic function $\phi _{0}(z)=\exp \left( -z^{\ast
}\alpha _{0} z\right) ,$ where $\alpha_{0}\geq \frac{I}{2}.$
Notice that $p_{\rho }(z)=\langle z|\rho |z\rangle $ if $\rho _{0}$
is the vacuum state corresponding to $\alpha _{0}=\frac{I}{2}$.

The function $p_{\rho }(z)$ is bounded by 1 and is a continuous probability
density, the normalization follows from the resolution of the identity operator in $\mathcal{H}$
\begin{equation*}
\int_{\mathbb{C}^{s}}D(z)\rho _{0} D(z)^{\ast }\frac{d^{2s}z}{\pi
^{s}}=I_{\mathcal{H}}.
\end{equation*}

\begin{proposition}
\label{T2} Let $f$ be a concave function on $[0,1],$ such that $f(0)=0,$ then for arbitrary state $\rho$
\begin{equation}
\int_{\mathbb{C}^{s}}f(p_{\rho }(z))\frac{d^{2s}z}{\pi ^{s}}\geq \int_{
\mathbb{C}^{s}}f(p_{|\zeta \rangle \langle \zeta |}(z))\frac{d^{2s}z}{\pi
^{s}}.  \label{aim}
\end{equation}
\end{proposition}

\begin{proof}
For any $c>0$ consider \textquotedblleft
measure-reprepare\textquotedblright\ channel $\Phi _{c}$ defined by the
relation
\begin{equation}
\Phi _{c}[\rho ]=\int \frac{d^{2s}z}{\pi ^{s}c^{2s}}\;\mbox{Tr}[\rho
D(c^{-1}z)\rho _{0}D^{\ast }(c^{-1}z)]\;D(z)\rho _{0}^{\prime}D^{\ast
}(z)\;,  \label{mr}
\end{equation}
where $\rho _{0}^\prime$ is another gauge-invariant Gaussian state with
the characteristic function $\phi _{0}^\prime(z) =\exp [-z^{\ast }\alpha _{0}^\prime z]$.
The map~(\ref{mr}) is a gauge-covariant bosonic Gaussian channel
which in the Heisenberg representation acts on $D(z)$ as
\begin{equation*}
\Phi _{c}^{\ast }[D(z)]=D(cz)\;\exp [-z^{\ast }(\alpha _{0}^\prime +c^{2}\alpha
_{0})z] ,
\end{equation*}
cf. \cite{P-2}. Therefore by Theorem \ref{T1},
\begin{equation}
\mathrm{Tr}f(\Phi _{c}[\rho ])\geq \mathrm{Tr}f(\Phi _{c}[|\zeta \rangle
\langle \zeta |])  \label{main}
\end{equation}
for all states $\rho $ and any coherent state $|\zeta \rangle \langle \zeta |
$. We will prove the Proposition \ref{T2} by taking the limit $c\rightarrow
\infty .$

In the proof we also use a simple generalization of the Berezin-Lieb
inequalities \cite{bl}:
\begin{equation}
\int_{\mathbb{C}^{s}}f(\underline{p}(z))\frac{d^{2s}z}{\pi ^{s}}\leq \mathrm{
Tr}f(\sigma )\leq \int_{\mathbb{C}^{s}}f(\bar{p}(z))\frac{d^{2s}z}{\pi ^{s}},
\label{blin}
\end{equation}
valid for any quantum state admitting the representation
\begin{equation*}
\sigma =\int_{\mathbb{C}^{s}}\underline{p}(z)D(z)\rho _{0}^{\prime}D(z)^{\ast }\frac{
d^{2s}z}{\pi ^{s}}
\end{equation*}
with a probability density $\underline{p}(z)$. In the right side of (\ref
{blin}) $$\bar{p}(z)=\mbox{Tr}\sigma \;D(z)\rho _{0}^{\prime}D^{\ast }(z).$$ The original inequalities refer to the
case where $\rho _{0}$ is a pure state, but the proof applies to the more
general case (see Appendix 2). In the inequalities (\ref{blin}) one has to
assume that $f$ is defined on $[0,\infty )$ (in fact, $\underline{p}(z)$ can
be unbounded). We shall assume this for a while.

Taking $\sigma =\Phi _{c}[\rho ],$ from (\ref{mr}) we have
\begin{equation*}
\underline{p}(z)=\frac{1}{c^{2s}}\;\mbox{Tr}\rho D(c^{-1}z)\rho
_{0}D^{\ast }(c^{-1}z)=\frac{1}{c^{2s}}\;p_{\rho} (c^{-1}z)\;.
\end{equation*}while%
\begin{equation}
\bar{p}(z)=\mathrm{Tr}\Phi _{c}[\rho ]D(z)\rho _{0}^{\prime}D(z)^{\ast }=\int_{
\mathbb{C}^{s}}\underline{p}(w)\mathrm{Tr}\rho _{0}^{\prime}D(z-w)\rho
_{0}^{\prime} D(z-w)^{\ast }\frac{d^{2s}w}{\pi ^{s}}.  \label{svert}
\end{equation}
By using the quantum Parceval formula \cite{hol}, we obtain
\begin{eqnarray*}
\pi ^{-s}\mathrm{Tr}\rho _{0}^{\prime}D(z)\rho _{0}^{\prime}D(z)^{\ast } &=&\int_{\mathbb{C}
^{s}}\phi _{0}^{\prime}(w)^{2}\exp \left( 2i\mathrm{Im}z^{\ast }w\right) \frac{d^{2s}w
}{\pi ^{2s}} \\
&=&\pi ^{-s}\det \left( 2\alpha _{0}^{\prime}\right) ^{-1}\exp \left( -\frac{1}{2}
z^{\ast }[\alpha _{0}^{\prime}]^{-1}z\right) \equiv q_{\alpha _{0}^{\prime}}(z)
\end{eqnarray*}
-- the probability density of a Gaussian distribution. Substituting this into (\ref{svert}), we have
\begin{eqnarray}
\bar{p}(z)
&=&\int d^{2s}w\;\underline{p}(w)\;
q_{\alpha_{0}^{\prime} }(z-w )  \notag  \\
&=&\int {d^{2s}w^{\prime }}\;p_{\rho} (w^{\prime })\;q_{\alpha_{0}^{\prime} }(z-cw^{\prime } ) \notag  \\
&=& \frac{1}{c^{2s}}p_{\rho }\ast q_{\alpha_{0}^{\prime} /{c^{2}}}(c^{-1}z).\label{gg}
\end{eqnarray}
Here $q_{\alpha_{0}^{\prime} /c^{2}}(z)=c^{2s}q_{\alpha_{0}^{\prime} }(cz)$ is the probability density of a Gaussian distribution
tending to $\delta -$function when $c\rightarrow \infty .$

With the change of the integration variable $c^{-1}z\rightarrow z$, the inequalities~(\ref{blin}) become
\begin{equation*}
\int_{\mathbb{C}^{s}}f(c^{-2s}p_{\rho }(z))\frac{d^{2s}z}{\pi ^{s}}\leq
c^{-2s}\mathrm{Tr}f(\Phi _{c}[\rho ])\leq \int_{\mathbb{C}
^{s}}f(c^{-2s}p_{\rho }\ast q_{\alpha_{0}^{\prime} /{c^{2}}}(z))\frac{d^{2s}z}{\pi ^{s}},
\end{equation*}
Substituting $
\rho =|\zeta \rangle \langle \zeta |,$ we have
\begin{equation*}
\int_{\mathbb{C}^{s}}f(c^{-2s}p_{|\zeta \rangle \langle \zeta |}(z))\frac{
d^{2s}z}{\pi ^{s}}\leq c^{-2s}\mathrm{Tr}f(\Phi _{c}[|\zeta \rangle \langle
\zeta |])\leq \int_{\mathbb{C}^{s}}f(c^{-2s}p_{|\zeta \rangle \langle \zeta
|}\ast\, q_{\alpha_{0}^{\prime} /{c^{2}}}(z))\frac{d^{2s}z}{\pi ^{s}}.
\end{equation*}
Combining the last two displayed formulas with (\ref{main}) we obtain
\begin{eqnarray}
&&\int_{\mathbb{C}^{s}}g(p_{\rho }(z))\frac{d^{2s}z}{\pi ^{s}}-\int_{\mathbb{
C}^{s}}g(p_{|\zeta \rangle \langle \zeta |}(z))\frac{d^{2s}z}{\pi ^{s}}
\notag \\
&\geq &\int_{\mathbb{C}^{s}}g(p_{\rho }(z))\frac{d^{2s}z}{\pi ^{s}}-\int_{
\mathbb{C}^{s}}g(p_{\rho }\ast q_{\alpha_{0}^{\prime} /{c^{2}}}(z))\frac{d^{2s}z}{\pi ^{s}}
,  \label{ine}
\end{eqnarray}
where we denoted $g(x)=f(c^{-2s}x),$ which is again a concave function.
Moreover, arbitrary concave polygonal function $g$ on $[0,1],$ satisfying $
g(0)=0,$ can be obtained in this way by defining
\begin{equation*}
f(x)=\left\{
\begin{array}{l}
g(c^{2s}x),\quad x\in \lbrack 0,c^{-2s}] \\
g(1)+g^{\prime }(1)(x-c^{-2s}),\quad x\in \lbrack c^{-2s},\infty )
\end{array}
\right. ,
\end{equation*}
hence (\ref{ine}) holds for any such function. Then the right hand side of
the inequality (\ref{ine}) tends to zero as $c\rightarrow \infty .$ Indeed,
for polygonal function $\left\vert g(x)-g(y)\right\vert \leq \varkappa
\left\vert x-y\right\vert ,$ and the asserted convergence follows from the
convergence $p_{\rho }\ast q_{\alpha_{0}^{\prime} /{c^{2}}}\longrightarrow p_{\rho }$ in $
L_{1}:$ if $p(z)$ is a bounded continuous probability density, then
\begin{equation*}
\lim_{c\rightarrow \infty }\int_{\mathbb{C}^{s}}\left\vert p\ast q_{\alpha_{0}^\prime /{
c^{2}}}(z)-p(z)\right\vert d^{2s}z=0.
\end{equation*}
Thus we obtain (\ref{aim}) for the concave polygonal functions $f.$ But for
arbitrary concave $f$ on $[0,1]$ there is a monotonously
nondecreasing sequence of concave polygonal functions $f_{n}$ converging to $
f$ . Applying Beppo-Levy's theorem, we obtain the statement.
\end{proof}

\section{Appendix}

1. The concatenation $\Phi =\Phi _{2}\circ \Phi _{1}$ of two Gaussian gauge-covariant channels $\Phi
_{1} $, $\Phi _{2}$ obeys the rule
\begin{eqnarray}
K &=&K_{2}K_{1},\quad  \label{c1} \\
\mu &=&K_{2}\mu _{1}K_{2}^{\ast }+\mu _{2}.  \label{c2}
\end{eqnarray}

\begin{proposition}
\label{prop1}\cite{P-2} Any bosonic Gaussian gauge-covariant channel $\Phi $ is a
concatenation of quantum-limited attenuator $\Phi _{1}$ and quantum-limited
amplifier $\Phi _{2}$.
\end{proposition}

\begin{proof}
By inserting
\begin{equation*}
\mu _{1}=\frac{1}{2}\left( I-K_{1}K_{1}^{\ast }\right) =\frac{1}{2}\left(
I-|K_{1}^{\ast }|^{2}\right) ,\quad \mu _{2}=\frac{1}{2}\left(
K_{2}K_{2}^{\ast }-I\right) =\frac{1}{2}\left( \left\vert K_{2}^{\ast
}\right\vert ^{2}-I\right)
\end{equation*}
into (\ref{c2}) and using (\ref{c1}) we obtain
\begin{equation}
\left\vert K_{2}^{\ast }\right\vert ^{2}=K_{2}K_{2}^{\ast }=\mu +\frac{1}{2}
(KK^{\ast }+I)\geq \left\{
\begin{array}{c}
I \\
KK^{\ast }
\end{array}
\right.  \label{ineq1}
\end{equation}
from the inequality (\ref{ineq}). By using operator monotonicity of the
square root, we have
\begin{equation*}
\left\vert K_{2}^{\ast }\right\vert \geq I,\quad \left\vert K_{2}^{\ast
}\right\vert \geq \left\vert K^{\ast }\right\vert .
\end{equation*}
The first inequality (\ref{ineq1}) implies that choosing
\begin{equation}
K_{2}=\left\vert K_{2}^{\ast }\right\vert =\sqrt{\mu +\frac{1}{2}(KK^{\ast
}+I)}  \label{K2}
\end{equation}
and the corresponding $\mu _{2}=\frac{1}{2}\left( \left\vert K_{2}^{\ast
}\right\vert ^{2}-I\right) ,$ we obtain (diagonalizable) quantum-limited
amplifier.

Then with
\begin{equation}
K_{1}=\left\vert K_{2}^{\ast }\right\vert ^{-1}K  \label{K1}
\end{equation}
we obtain, taking into account the second inequality in (\ref{ineq1})
\begin{equation}
K_{1}^{\ast }K_{1}=K^{\ast }\left\vert K_{2}^{\ast }\right\vert
^{-2}K=K^{\ast }\left[ \mu +\frac{1}{2}(KK^{\ast }+I)\right] ^{-1}K\leq I,  \label{K3}
\end{equation}
which implies $K_{1}^{\ast
}K_{1}\leq I,$ hence $K_{1}$ with the corresponding $\mu _{1}=\frac{1}{2}
\left( I-K_{1}K_{1}^{\ast }\right) $ give the quantum-limited attenuator.
\end{proof}

\begin{remark}
\label{rem} The inequality (\ref{mugr}) via (\ref{K3}) implies $K_{1}^{\ast
}K_{1}<I.$ Invertibility of $K$ implies $K_{1}^{\ast }K_{1}>0$.
\end{remark}

2. For completeness we sketch the proof of the required generalization of
Berezin-Lieb inequalities. Let $\mathcal{X}$ be a measurable space with $
\sigma -$finite measure $\mu ,$ and $P(x)$ a weakly measurable function on $
\mathcal{X}$ \ whose values are density operators in a separable Hilbert
space $\mathcal{H}$ such that
\begin{equation*}
\int_{\mathcal{X}}P(x)\mu (dx)=I_{\mathcal{H}},
\end{equation*}
where the integral converges in the sense of weak operator topology. Let $
\rho $ be a density operator in $\mathcal{H}$ admitting representation
\begin{equation*}
\rho =\int_{\mathcal{X}}\underline{p}(x)P(x)\mu (dx),
\end{equation*}
where $\underline{p}(x)$ is a bounded probability density. Denote $\bar{p}
(x)=\mathrm{Tr}\rho P(x),$ which is a probability density uniformly bounded
by 1. Then for a concave function $f$ defined on $[0,\infty )$ and
satisfying $f(0)=0$
\begin{equation}
\int_{\mathcal{X}}f(\underline{p}(x))\mu (dx)\leq \mathrm{Tr}f(\rho )\leq
\int_{\mathcal{X}}f(\bar{p}(x))\mu (dx).  \label{be}
\end{equation}

Put $k=\max \left\{ 1,\sup_{x}\underline{p}(x)\right\} $ and consider
restriction of $f$ to $[0,k].$ Then there is a monotonously nondecreasing
sequence of concave polygonal functions $f_{n}$ converging to $f$ pointwise
on $[0,k]$ and satisfying $\ f_{n}(0)=0.$ Since $\left\vert
f_{n}(x)\right\vert \leq \varkappa _{n}\left\vert x\right\vert ,$ the
integrals and the trace in (\ref{be}) with $f$ replaced by $f_{n}$ are
finite for all $n.$ Let us prove (\ref{be}) for concave polygonal functions $
f_{n}$ and then take the limit $n\rightarrow \infty .$ This will also show
that the integrals and trace in (\ref{be}) are well defined although may
take the value $+\infty .$

The second inequality follows from $\mathrm{Tr}f(\rho )P(x)\leq f(\mathrm{Tr}
\rho P(x))$ which is a consequence of Jensen inequality applied along with
the spectral decomposition of $\rho .$ To prove the first inequality
consider the positive operator-valued measure
\begin{equation*}
M(B)=\int_{B}P(x)\mu (dx),\quad B\subseteq \mathcal{X},
\end{equation*}
and its Naimark dilation to a projection-valued measure $\left\{
E(B)\right\} $ in a larger Hilbert space $\mathcal{\tilde{H}}\supseteq
\mathcal{H}.$ Consider the bounded operator $R=\int_{\mathcal{X}}\underline{p
}(x)E(dx)$ in $\mathcal{\tilde{H}},$ then $f(R)=\int_{\mathcal{X}}f(
\underline{p}(x))E(dx)$ and
\begin{equation*}
\rho =PRP,\quad Pf(R)P=\int_{\mathcal{X}}f(\underline{p}(x))\mu (dx),
\end{equation*}
where $P$ is projection from $\mathcal{\tilde{H}}$ onto $\mathcal{H}$. The
required inequality then follows from the more general fact $\mathrm{Tr}
Pf(R)P\leq \mathrm{Tr}f(PRP)$ \cite{nai}.

\section{Addendum}
K. R. Parthasarathy \cite{par} pointed out to
us that attainability of an infimum in Lemma 1 of \cite{P-1} needs
explanation. A proof answering this question was published recently in \cite{UMN}.
Below we give still shorter argument using notations from  \cite{P-1}.

Let  $f$ \ be a strictly concave function on $[0,1]$ such that $f(0)=0,$
then define%
\begin{equation}
F(\sigma )=\mathrm{Tr}f(\sigma ),  \label{F}
\end{equation}%
for a density operator $\sigma .$ It is sufficient to show that if $\mathcal{%
A}_{\kappa }$ is quantum-limited amplifier then
\begin{equation}
F(\mathcal{A}_{\kappa }[\rho ])\geq F(\mathcal{A}_{\kappa }[|0\rangle
\langle 0|]),  \label{main2}
\end{equation}%
for all quantum states $\rho $ and for the vacuum state $|0\rangle \langle
0|.$

This proof is based on the following observation of A. Mari. The original proof of Lemma 1  in \cite{P-1} still works
when restricting the input $\rho $ to the set of states $\mathfrak{S}_{n}$
supported by the subspace spanned by the first $n$ Fock vectors. This set is
finite-dimensional (the infimum is hence attained), and it is mapped into
itself by a quantum-limited attenuator $\mathcal{E}_{\eta }$ (see e.g. \cite{ivan}). Therefore,
the  proof given in  \cite{P-1} implies that (\ref{main2}) holds for $\rho
\in \mathfrak{S}_{n}$. Notice that $\mathfrak{S}_{n}$ is dense in the set of
all quantum states because for a given state $\rho $ the states $\rho _{n}=%
\frac{P_{n}\rho P_{n}}{\mathrm{Tr}P_{n}\rho P_{n}}\in \mathfrak{S}_{n},$
where $P_{n}$ is the projection on the subspace spanned by the first $n$
Fock vectors, converge to $\rho $ in trace norm.

Then we use the approximation argument from the proof of Theorem 1 of the present paper. First notice that $%
f(x)\geq f(1)x,$ hence the right-hand side of (\ref{F}) is unambiguously
defined with $F(\sigma )\geq $ $f(1)$. Moreover, the functional $F(\sigma ),$
and hence $F(\Phi \lbrack \rho ]),$ is a pointwise limit of monotonously
nondecreasing sequence of continuous functionals. In fact, for arbitrary
concave $f$ on $[0,1]$ there is a monotonously nondecreasing sequence of
concave polygonal functions $f_{m}$ converging to $f$ pointwise. Any such
function is Lipschitz, $|f_{m}(x)-f_{m}(y)|\leq \varkappa _{m}|x-y|$, hence
the corresponding functional $F_{m}(\sigma )=\mathrm{Tr\,}f_{m}(\sigma )$ is
continuous (Lemma below).  Hence
\begin{equation*}
F_{m}(\Phi \lbrack \rho ])=\lim_{n\rightarrow \infty }F_{m}(\Phi \lbrack
\rho _{n}])\geq F_{m}(\Phi \lbrack |0\rangle \langle 0|]).
\end{equation*}%
Taking the limit $m\rightarrow \infty $ of monotonously nondecreasing
sequence $F_{m}$ gives (\ref{main2}) for a given state $\rho .$

\textbf{Lemma}
Let $|f(x)-f(y)|\leq \varkappa |x-y|$ for $x,y\in \lbrack 0,1],$
then $\left\vert \mathrm{Tr\,}f(\sigma )-\mathrm{Tr\,}f(\rho )\right\vert
\leq \varkappa \left\Vert \sigma -\rho \right\Vert _{1}.$

\begin{proof}
Let $\lambda _{1}\geq \lambda _{2}\geq ...\left( \mu _{1}\geq \mu _{2}\geq
...\right) $ be the eigenvalues of $\sigma $ (correspondingly, $\rho $). By
Mirsky's theorem (Lemma IV.3.2 \cite{P})
\begin{equation}
\sum_{i}\left\vert \lambda _{i}-\mu _{i}\right\vert \leq \left\Vert \sigma
-\rho \right\Vert _{1}.  \label{fann3}
\end{equation}%
Therefore$\left\vert \mathrm{Tr\,}f(\sigma )-\mathrm{Tr\,}f(\rho
)\right\vert \leq \sum_{i}\left\vert f(\lambda _{i})-f(\mu _{i})\right\vert
\leq \varkappa \sum_{i}\left\vert \lambda _{i}-\mu _{i}\right\vert \leq
\varkappa \left\Vert \sigma -\rho \right\Vert _{1}.$
\end{proof}

\section{Acknowledgments}
The authors are grateful to M. E. Shirokov for discussion. The work of A.S. Holevo was supported by the grant of Russian Scientific Foundation  (project No 14-21-00162).

\end{document}